\documentclass[a4paper,UKenglish,cleveref, autoref, thm-restate]{lipics-v2021}

\usepackage{svg}
\usepackage{tcolorbox}
\usepackage{csquotes}
\usepackage{todonotes}

\usepackage{thmtools} 
\usepackage{thm-restate}

\newtheorem{question}{Question}

\title{On $2$-Layer $k$-Matching-Planar Graphs} %TODO Please add

 \author{Saeed Odak}{Aalto University, Finland}{saeed.odak@aalto.fi}{https://orcid.org/0009-0005-6290-0965}{}

 \author{Jonathan Rollin}{FernUniversität in Hagen, Germany}{jonathan.rollin@fernuni-hagen.de}{https://orcid.org/0000-0002-6769-7098}{}

 \author{Torben Scheele}{TU Dortmund, Germany}{torben.scheele@tu-dortmund.de}{https://orcid.org/0009-0006-6119-6598}{Funded by the Deutsche Forschungsgemeinschaft (DFG, German Research Foundation)  -- project number 550144388.}

\authorrunning{S.~Odak, J.~Rollin, T.~Scheele}

\Copyright{Saeed Odak, Jonathan Rollin, Torben Scheele}

 %TODO mandatory: Please choose ACM 2012 classifications from https://dl.acm.org/ccs/ccs_flat.cfm 

 %TODO mandatory; please add comma-separated list of keywords

\hideLIPIcs

\acknowledgements{This work was initiated at the Rhein-Ruhr Computational Geometry Workshop RRW 25. We thank the organizers and the participants for the fruitful atmosphere.}%optional

\nolinenumbers %uncomment to disable line numbering

\EventEditors{}
\EventNoEds{0}
\EventLongTitle{}
\EventShortTitle{}
\EventAcronym{}
\EventYear{}
\EventDate{}
\EventLocation{}
\EventLogo{}
\SeriesVolume{}
\ArticleNo{}
\begin{document}

\maketitle

\begin{abstract}
A graph is \emph{$k$-matching-planar} if it admits a drawing in the plane such that, for every edge $e$, the edges crossing $e$ contain no matching of size greater than $k$. The class of $k$-matching-planar graphs generalizes other beyond-planar graph classes, such as $k$-planar and fan-planar graphs. In a \emph{$2$-layer drawing} of a bipartite graph, the vertices of the two bipartition classes are placed on two parallel horizontal lines and edges are drawn as straight-line segments between them.

We prove that every graph with a $2$-layer $k$-matching-planar drawing has pathwidth at most $2k+1$. Moreover, for every $k \ge 0$, we construct a graph with a $2$-layer $k$-matching-planar drawing whose pathwidth is $3\lfloor k/2\rfloor + 1$. On the algorithmic side, we consider the one-sided recognition problem where a fixed embedding of the vertices on one side is given.
We show that this problem is NP-hard. On the other hand, we prove that the problem is fixed-parameter tractable with respect to~$k$. Finally, we prove that the two-sided variant cannot be approximated within any constant factor in polynomial time unless~$\operatorname{P}=\operatorname{NP}$.

\end{abstract}

\section{Introduction}

A drawing of a graph in the plane is \emph{planar} if no two edges cross. As planarity is restrictive (e.g., every planar $n$-vertex graph has at most $3n-6$ edges), the graph drawing community has long studied \emph{beyond-planar} graph classes, where edges may cross but the crossings are controlled in some way~\cite{DLM19,hong2020beyond}. A central thread in this area is to identify the broadest possible classes that still inherit the desirable structural properties of planar graphs.

\subparagraph*{$k$-Matching-Planar Graphs.} Among the most natural beyond-planar restrictions is to bound the number of crossings per edge. For an integer $k\geq 0$, a drawing is $k$-planar if every edge is involved in at most $k$ crossings. The class of $k$-planar graphs is classical and extensively studied~\cite{grigoriev2007algorithms, hoffmann2020simple, pach1997graphs}.

A substantially more general notion was recently considered by Hendrey, Karol, and Wood~\cite{HKW25}. A drawing $G$ is \emph{$k$-matching-planar} if, for every edge $e\in E(G)$, the matching number of the set of edges crossing $e$ is at most $k$. Equivalently, no edge is crossed by $k+1$ pairwise non-adjacent edges.\footnote{An equivalent concept (under the name \emph{$k$-independent crossing graphs}) was considered in passing by Merker, Scherzer, Schneider, and Ueckerdt~\cite{MSSU24} in the context of intersection graphs.} Every $k$-planar graph is $k$-matching-planar, but the converse fails dramatically: in a $k$-matching-planar drawing a single edge may participate in arbitrarily many crossings, provided the crossing edges share endpoints. For instance, $K_{3,n}$ is $1$-matching-planar for every $n$, yet in any drawing of $K_{3,n}$ some edge is crossed $\Omega(n)$ times since $K_{3,n}$ has crossing number $\Omega(n^2)$ \cite{HKW25,kleitman1970crossing}.

The class of $k$-matching-planar graphs subsumes many beyond-planar classes that have been studied separately. For example, the authors in \cite{HKW25} argue that the beyond-planar graph classes of \emph{fan-planar} graphs~\cite{KU22} and \emph{fan-crossing} graphs~\cite{Brandenburg20} are $1$-matching-planar. Moreover, they show that every \emph{$k$-fan-bundle-planar} graph defined by~\cite{ABKKS18} is $2k$-matching-planar.

Hendrey et al. ~\cite{HKW25} demonstrated that the class of $k$-matching-planar graphs is the broadest beyond-planar class for which a \emph{product structure theorem} is currently known. Therefore this class enjoys bounded queue-number, bounded nonrepetitive chromatic number, polynomial $p$-centered chromatic numbers, bounded boxicity, and asymptotic dimension~$2$. Studying the global structure of $k$-matching-planar graphs is therefore well-motivated: any structural result for this class automatically applies to all of the more restrictive classes mentioned~above.

\subparagraph*{$2$-Layered Drawings.} Hierarchical data relationships are often presented as layered graph drawings, where the vertices are placed on several consecutive horizontal lines. An important case are \emph{$2$-layer drawings} of bipartite graphs. A \emph{$2$-layer drawing} of a bipartite graph $G=(X\cup Y, E)$ places the vertices of $X$ and $Y$ on two parallel horizontal lines and draws each edge as a straight-line segment between the two lines. This drawing model is one of the oldest and most natural ways to visualize bipartite graphs~\cite{HN13}. A typical drawing~objective, to improve readability, is the minimization of the number of crossings between the edges. %, which has been studied intensively and 
Two-layer crossing minimization is NP-hard, when the order of one layer is fixed~\cite{EWh94,EWo94,GJ83,HN13}. More recently, the local version of minimizing the maximum number of crossings per edge was studied~\cite{angelini20212, ALFS24, di20142, OneSidedKPlanar, KOW25}.

The bipartite graphs that admit a crossing-free $2$-layer drawing are exactly the forests of caterpillars, which have pathwidth at most~$1$~\cite{EWh94}. This characterization was extended by Wood~\cite{Wood23}, who showed that bipartite graphs of bounded pathwidth are precisely those admitting a $2$-layer drawing with a suitable forbidden crossing pattern.
Once one allows crossings, an extensive literature has investigated which bipartite graphs admit $2$-layer drawings with various restrictions on the crossings, what the maximum density of such graphs is, which algorithms can decide membership, and which graph parameters (such as pathwidth and treewidth) are bounded. We summarize the most relevant prior results below.

For \emph{$2$-layer $k$-planar graphs}, those admitting a $2$-layer drawing with at most $k$ crossings per edge, Angelini, Da Lozzo, F\"orster, and Schneck~\cite{ALFS24} initiated a systematic study. Among other structural properties, they showed that every $2$-layer $k$-planar graph has pathwidth at most $k+1$, with a lower bound of $(k+3)/2$. Okada~\cite{Okada25} subsequently sharpened the lower bound by exhibiting for every $k\geq 0$, a $2$-layer $k$-planar graph of pathwidth exactly~$k+1$. %, thereby resolving the pathwidth of $2$-layer $k$-planar graphs. 
Recognizing $2$-layer $k$-planar graphs is $\operatorname{XNLP}$-complete and hence W[$t$]-hard for every~$t$, however Kobayashi, Okada, and Wolff~\cite{KOW25} gave an $\operatorname{XP}$-algorithm parameterised by $k$ as well as an FPT-algorithm for the one-sided variant where the order of one layer is given. It has later been shown that the one-sided variant is NP-complete~\cite{OneSidedKPlanar}.

For \emph{$2$-layer fan-planar graphs}, those in which the edges crossing any single edge share a common endpoint, Binucci et al.~\cite{BCDGKKMT17} characterized the recognizable graphs as subgraphs of \emph{stegosaurus graphs} and gave a linear-time recognition algorithm for biconnected bipartite graphs. Polynomial-time recognition for general bipartite graphs was an open problem for several years and was only recently resolved by Kobayashi and Okada~\cite{KO25}. \emph{Proper $h$-layer fan-planar drawings} were studied by Biedl et al.~\cite{BCKMNR20}, who proved that such graphs have bounded pathwidth and that the corresponding recognition problem is FPT in $h$ when the embedding is fixed.\footnote{Layered drawings where each edge has its endpoints on adjacent layers are called proper in~\cite{BCKMNR20}.}

\subparagraph*{Our Contribution.} We introduce and study the class of \emph{$2$-layer $k$-matching-planar graphs}, that is, bipartite graphs admitting a $2$-layer drawing in which, for every edge $e$, the matching number of the set of edges crossing $e$ is at most $k$.

Our results address pathwidth, recognition complexity, and a parameterized~algorithm:

\begin{itemize}
    \item \textbf{Pathwidth.} Every graph with a $2$-layer $k$-matching-planar drawing has pathwidth at most $2k+1$. For every $k\geq 0$, there is a graph $G_k$ with a $2$-layer $k$-matching-planar drawing whose pathwidth equals $3\lfloor k/2\rfloor+1$. In particular, the pathwidth of $2$-layer $k$-matching-planar graphs is $\Theta(k)$, but unlike in the $2$-layer $k$-planar setting, the constant in front of $k$ is strictly greater than~$1$.
    \item \textbf{NP-completeness of One-Sided Recognition.} Deciding whether a bipartite graph with one given vertex order admits a $2$-layer $k$-matching-planar drawing is NP-complete.
    \item \textbf{FPT Recognition.} The one-sided recognition problem is fixed-parameter tractable in $k$.
    \item \textbf{Inapproximability of Two-Sided Recognition.} The optimization version of the two-sided recognition problem, asking for the smallest $k$ such that a given bipartite graph admits a $2$-layer $k$-matching-planar drawing, cannot be approximated within any constant factor in polynomial time unless $\operatorname{P}=\operatorname{NP}$, even for trees.
\end{itemize}

Together, these results give a qualitative picture of the pathwidth and one-sided recognition complexity of $2$-layer $k$-matching-planar graphs, and generalize the corresponding results of Angelini et al.~\cite{ALFS24}, Okada~\cite{Okada25}, and Kobayashi et al.~\cite{KOW25} to $2$-layer $k$-matching-planar graphs.

\subparagraph*{Technical Overview.} At first sight, one might hope to derive our results from the existing $2$-layer $k$-planar arguments by simply replacing ``at most $k$ crossings per edge'' by ``at most $k$ edges in any matching''. Several genuine obstacles prevent this. We require new structural tools aligned with the ``matching'' property, which is different from the ``number of crossings'' property. 

The pathwidth proof of Angelini et al.~\cite{ALFS24} for $2$-layer $k$-planar graphs proceeds by ordering the edges along the drawing and, for each edge $e_i$, using directly the at most $k$ edges crossing $e_i$ to populate a bag. In the matching-planar setting this is unavailable: an edge may be crossed by an unbounded number of edges, and the ``at most $k$'' guarantee delivers only a bounded \emph{matching} among the crossing edges. To compensate, we work with vertex covers instead of crossing edges. By K\H{o}nig's theorem, every bipartite set of edges of matching number at most $k$ admits a vertex cover of size at most $k$, but a fixed vertex cover does not suffice: a single bag based on one vertex cover would not be consistent across the drawing, and the vertex covers chosen for adjacent edges could differ wildly. Our proof therefore relies on a canonical choice of minimum vertex covers along the drawing, developed by Jeong, S{\ae}ther, and Telle~\cite{JST18} that is monotone with respect to neighboring edges. Our resulting bound on the pathwidth is by a factor of two larger than the bound $k+1$ for $2$-layer $k$-planar graphs~\cite{ALFS24}. Our lower bound shows that this loss is at least partially unavoidable: there are $k$-matching-planar graphs with pathwidth at least~$3k/2$.

The NP-hardness reduction for \textsc{OneSided-$k$-MatchingPlanarity} is inspired by the reduction in~\cite{OneSidedKPlanar} for the $k$-planar variant, but the additional flexibility of matching-planarity forces a more elaborate construction. In a $k$-planar drawing, an edge with $k$ crossings is already saturated; we can use such an edge to anchor positions in the construction. In a $k$-matching-planar drawing, the same edge may still be crossed by many further edges as long as they are pairwise adjacent, so we cannot prevent additional crossings simply by reaching the local crossing budget. We circumvent this by introducing additional gadgets that consist of various complete bipartite graphs serving as ``anchors''.

Our FPT-algorithm follows the dynamic-programming template of Kobayashi, Okada, and Wolff~\cite{KOW25} for $2$-layer $k$-planar graphs.  There the authors describe reduction procedures to yield an instance with two crucial properties: (i) every vertex has degree at most $k+1$, and (ii) edges crossing a fixed window have endpoints contained in a small region. 
For the matching-planar setting, we develop two new structural lemmas. Lemma~\ref{lem:redundant} (the \emph{twin-reduction} lemma) shows that whenever a vertex $u$ has more than $2|N(u)|$ ``twin'' vertices with the same neighborhood, one can be removed without changing $k$-matching-planarity. Lemma~\ref{lem:windowsize} (the \emph{window-size} lemma) then shows that, after exhaustive twin-reduction, a window of at most $\ell=2^{O(k)}$ vertices in $X$ accommodates all relevant interactions. Once these structural facts are in hand, the dynamic programming over windows of size $\ell$ proceeds along the lines of~\cite{KOW25}, although the bookkeeping is more involved due to the larger windows.

Finally, for the two-sided variant we follow the reduction strategy from the \textsc{Bandwidth} problem on trees used by~\cite{KOW25} for the $k$-planar setting. Given a tree $T$, we construct an auxiliary tree $G_T$ as the $1$-subdivision of $T$ with an additional pendant attached to each original vertex. We show that the smallest $k$ for which $G_T$ admits a $2$-layer $k$-matching-planar drawing is sandwiched between $\operatorname{bw}(T)/2$ and $3\operatorname{bw}(T)$ (up to additive constants). Combined with the constant-factor inapproximability of bandwidth on trees~\cite{DFU11}, this transfers to constant-factor inapproximability of the smallest matching-planarity parameter.

\section{Preliminaries}
 
\subparagraph*{Drawings and Crossings.} A \emph{drawing} $G$ is a graph with $V(G)\subset \mathbb{R}^2$ in which each edge $uv\in E(G)$ is a non-self-intersecting curve in $\mathbb{R}^2$ with endpoints $u$ and $v$ that does not contain any vertex of $G$ in its interior. A drawing $G$ is \emph{simple} if any two edges share at most one point, including endpoints. A simple drawing $G$ is:
\begin{itemize}
    \item \emph{planar} if no edge of $G$ is involved in a crossing,
    \item \emph{$k$-planar} if each edge of $G$ is involved in at most $k$ crossings, and
    \item \emph{$k$-matching-planar} if no edge of $G$ is crossed by $k+1$ pairwise non-adjacent edges.
\end{itemize}
A graph belongs to the class of planar, $k$-planar, or $k$-matching-planar graphs if it admits a simple embedding that is respectively planar, $k$-planar, or $k$-matching-planar.

\subparagraph*{Pathwidth.} A \emph{path decomposition} of a graph $G=(V,E)$ is a sequence $B_1,\ldots,B_t$ of subsets of $V$, called \emph{bags}, satisfying the following three conditions:
\begin{enumerate}
    \item Each vertex lies in some bag: $\bigcup_{1\leq i\leq t}B_i=V$.
    \item Each edge lies in some bag: for every edge $uv\in E$ there is an index $i$ with $u,v\in B_i$.
    \item For every vertex $v\in V$, the set $\{i\colon v\in B_i\}$ is a contiguous interval; equivalently, $B_i\cap B_k \subseteq B_j$ whenever $1\leq i\leq j\leq k\leq t$.
\end{enumerate}
The \emph{width} of a path decomposition is $\max_{1\leq i\leq t} |B_i|-1$, and the \emph{pathwidth} $\operatorname{pw}(G)$ of $G$ is the minimum width over all path decompositions of $G$. Pathwidth is an important parameter in graph structure theory and parameterized complexity~\cite{CFKLMPPS15,DHK05,RS83}, with many connections to graph drawing~\cite{DFKLMNRRWW08,Hlineny03,Suderman04,Zehavi22}.
 
\subparagraph*{Bipartite Graphs and $2$-Layer Drawings.} We consider bipartite graphs $G=(X\cup Y,E)$ with bipartition classes $X$ and $Y$. The subgraph of $G$ induced by a vertex set $V\subseteq X\cup Y$ is denoted by $G[V]$. The neighborhood of a vertex $u$ is denoted by $N(u)$, and for a vertex set $U$, we write $N(U)=\bigcup_{u\in U}N(u)$.
 
In a \emph{$2$-layer drawing} of $G$, the bipartition classes $X$ and $Y$ are placed on two distinct parallel horizontal lines, and edges are drawn as straight-line segments between their endpoints. We let $\prec_X$ and $\prec_Y$ denote the linear orders of $X$ and $Y$, respectively, induced by reading the vertices from left to right along their layer. We also describe these orders informally by saying that a vertex $u$ lies \emph{to the left} or \emph{to the right} of another vertex $v$ from the same bipartition class. In such a drawing, two edges $uv$ and $u'v'$, with $u,u'\in X$ and $v,v'\in Y$, cross if and only if $u\prec_X u'$ and $v'\prec_Y v$, or vice versa. A \emph{$2$-layer $k$-matching-planar drawing} is a $2$-layer drawing which is $k$-matching-planar. A bipartite graph is \emph{$2$-layer $k$-matching-planar} if it admits such a drawing.

\section{Pathwidth of $2$-Layer $k$-Matching-Planar Graphs}
\label{sec:pathwidth}

\subparagraph*{An Upper Bound on the Pathwidth.} Wood~\cite{Wood23} established a tight relation between $2$-layer drawings of bipartite graphs and their pathwidth: A class $\mathcal{G}$ of bipartite graphs has bounded pathwidth if and only if there exist $p$, $q$, $r\in\mathbb{N}$ such that every graph in $\mathcal{G}$ has a $2$-layer drawing with no $p$-crossing and no $(q, r)$-crossing. Here, a $p$-crossing is formed by $p$ pairwise crossing edges and a $(q,r)$-crossing is formed by two non-crossing matchings $M_1$ and $M_2$ of sizes $q$ and $r$, respectively, such that each edge of $M_1$ crosses each edge of~$M_2$. This implies that bipartite graphs with a $2$-layer $k$-matching-planar drawing have pathwidth~$O(k^3)$. We improve this bound by showing that every graph with a $2$-layer $k$-matching-planar drawing has a path decomposition of width at most $2k+1$. Let us remark that no converse statements to our results hold: for each $k$, any sufficiently large 1-subdivided star has pathwidth 2 and is not $2$-layer $k$-matching-planar.

\begin{theorem}
    The pathwidth of any $2$-layer $k$-matching-planar graph is at most $2k+1$.
\end{theorem}

\begin{proof}
Let $\Gamma$ be a $2$-layer $k$-matching-planar drawing of some graph $G$.
Without loss of generality, we assume that $\Gamma$ is maximal (so any new straight edge between non-adjacent vertices on opposite sides would violate $k$-matching-planarity).
Let $H$ denote an edge-maximal plane subgraph in the drawing $\Gamma$ that contains the edge between the two leftmost vertices, contains the edge between the two rightmost vertices, and contains no isolated vertices.
These edges on the left and right boundary exist in $\Gamma$ due to the maximality of $\Gamma$.
Observe that $H$ consists of disjoint caterpillars. 
Let $e_1,\ldots,e_m$ denote the edges in $H$ from left to right, that is, $uv$ is before $u'v'$ if either $v$ is left of $v'$ or $v=v'$ and $u$ is left of $u'$.
Let $S_i$ denote the set of edges crossing $e_i$ in $\Gamma$.
We will choose a special vertex cover for each $S_i$ as follows.
Let $L_i$ denote the set of vertices to the left of $e_i$ and let $R_i$ denote the set of vertices to the right of $e_i$ (both excluding the endpoints of  $e_i$).
Then $(L_i\cup R_i,S_i)$ is a bipartite graph.
Let $\mathcal{F}_i$ denote the family of all minimum vertex covers of this graph.
We define a \emph{left-K\H{o}nig-cover} as the set $C_i$ that consists of all vertices from $L_i$ that are contained in the union of the vertex covers in $\mathcal{F}_i$ and of all vertices from $R_i$ that are contained in the intersection of all the vertex covers in $\mathcal{F}_i$ (see Figure \ref{fig:upperbound}). That is

\[C_i = \left(\bigcup_{C \in \mathcal{F}_i} L_i 
\cap C\right) \cup 
\left(\bigcap_{C \in \mathcal{F}_i} R_i 
\cap C\right).\]
Jeong, S{\ae}ther, and Telle~\cite[Corollary 3.3 and Lemma 3.4]{JST18} proved the following properties.\footnote{In their notation, $C_i$ is an $L_i$-K\H{o}nig cover of the bipartite graph $(L_i\cup R_i,S_i)$.}

\begin{figure}[htbp]
  \centering
  \includegraphics[]{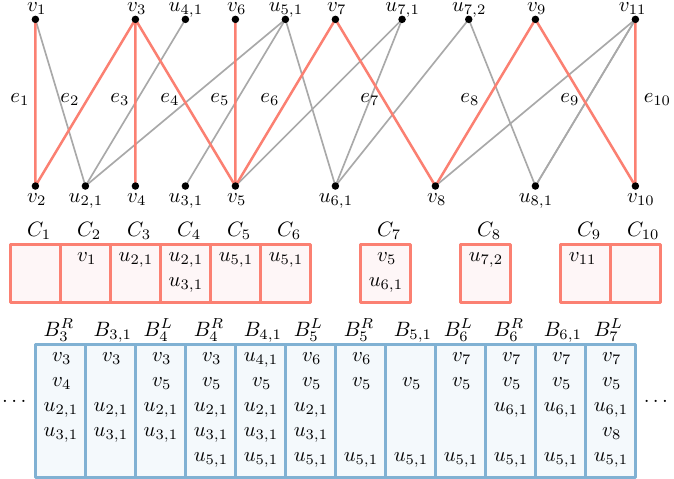}
  \caption{A $2$-layer 2-matching-planar drawing with an edge-maximal plane subgraph (red) as well as the left-K\H{o}nig-cover for its edges and the resulting path decomposition of width at most $5$.}
  \label{fig:upperbound}
\end{figure}

\begin{lemma}[\cite{JST18}]
\label[lemma]{lem:Koenig-cover}
\begin{enumerate}[(I)]
	\item For each $i$, $1\leq i\leq m$, $C_i$ is a minimum vertex cover of $S_i$.
	
	\item For each $i$, $1\leq i<m$, we have $C_i\cap L_i \supseteq C_{i+1}\cap L_i$ and $C_i \cap R_{i+1} \subseteq C_{i+1} \cap R_{i+1}$.
\end{enumerate}	
\end{lemma}

As $\Gamma$ is $k$-matching-planar, \Cref{lem:Koenig-cover}~(I) implies that $C_i$ has size at most $k$.
Now we are ready to describe a path-decomposition of $G$ (see again Figure \ref{fig:upperbound}).
\begin{itemize}
	\item Let $B^R_1 = e_1 \cup C_1 \cup C_{2}$ and let $B^L_{m} = e_m \cup C_{m-1} \cup C_m$.
	
	\item For each $i$, $1< i< m$, let $B^L_i= e_i \cup C_{i-1} \cup C_i$ and let $B^R_i= e_i \cup C_i \cup C_{i+1}$.
	
	\item For each $i$, $1\leq i < m$, let $u_{i,1},\ldots,u_{i,t_i}$ denote the vertices between $e_i$ and $e_{i+1}$ in an arbitrary order.
	For each $j$, $1\leq j\leq t_i$, let $B_{i,j}= \{u_{i,j}\} \cup (e_i\cap e_{i+1}) \cup C_i \cup C_{i+1}$.
\end{itemize}
Each bag $B_i^L$, $B_i^R$, and $B_{i,j}$ has size at most $2k+2$.
We claim that 
\[B^R_1,B_{1,1},\ldots,B_{1,t_1}, B^L_2, B^R_2 ,B_{2,1}, \ldots, B_{2,t_2},B^L_3, B^R_3,B_{3,1}, \ldots, B^L_{m}\]
is a path decomposition of $G$.

\begin{description}
	\item[Each vertex in some bag:] Due to the technical assumptions on $H$, each vertex from $G$ is either incident to an edge from $H$ or between two edges from $H$.
	Each vertex incident to an edge from $H$ is in some bag $B^L_i$ or $B^R_i$. Any other vertex is in some bag $B_{i,j}$.
	
	\item[Each edge in some bag:] We have $e_1\subseteq B^R_1$, $e_m\subseteq B^L_m$ and, for $1<i<m$, $e_i\subseteq B^L_i$.
	Any other edge $uv$ crosses some edge from $H$, as $H$ is edge-maximal plane.
	Let $e_\ell$ and $e_{r}$ denote the leftmost and rightmost edges from $H$ crossed by $uv$.
	Then $uv$ crosses each edge $e_i$ in $H$ with $\ell\leq i\leq r$.
	For each $i$, $\ell\leq i\leq r$, $uv$ is in $S_i$ and so the left-K\H{o}nig cover $C_i$ of $S_i$ contains at least one of $u$ or $v$.
	Assume that $u$ is to the left of $e_\ell$.
	If $v\in C_\ell$ and $u$ is an endpoint of $e_{\ell-1}$, then $u$, $v\in B^R_{\ell-1}$.
	If $v\in C_\ell$ and $u$ is between $e_{\ell-1}$ and $e_\ell$, then $u$, $v\in B_{\ell-1,j}$ for some $j$.
	Similarly, if $u\in C_r$,  then $u$, $v\in B^L_{r+1}$ or $u$, $v\in B_{r,j}$ for some $j$.
	If neither $v\in C_\ell$ nor $u\in C_r$, then there is some $i$, $\ell \leq i < m$, such that $u$, $v\in C_{i}\cup C_{i+1}$.
	Hence, $u$, $v\in B^R_i$. See \cref{fig:every_edge_cases} for an illustration of the different cases.

    \begin{figure}[htbp]
      \centering
      \includegraphics[]{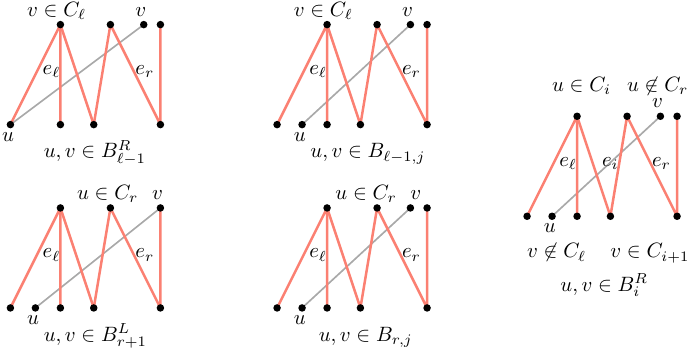}
      \caption{Examples for the different cases considered in the proof that every edge is in some bag.}
      \label{fig:every_edge_cases}
    \end{figure}

    \item[Each vertex's bags induce a subpath:] We distinguish between vertices based on their inclusion in the maximal plane subgraph $H$ and the left-K\H{o}nig covers.

    \begin{itemize}
        \item \textbf{Vertices not in $H$ and not in any cover.} 
        Let $u$ be such a vertex. Since $u \notin V(H) \cup \bigcup C_i$, $u$ appears exclusively in the bag $B_{i,j}$ corresponding to its specific position. %A single bag forms a trivial subpath.
        
        \item \textbf{Vertices not in $H$, but in some cover.}
        Let $u$ be a vertex not incident to any edge in $H$. Therefore, $u$ is located strictly between two edges $e_{j}$ and $e_{j+1}$.
        For any $i > j$, $u$ lies to the left of $e_i$ ($u \in L_i$). If $u \in C_i$, then by \Cref{lem:Koenig-cover}~(II) we have $u \in C_{i-1}$ down to $C_{j+1}$.
        Symmetrically, for any $i \le j$, $u$ lies to the right of $e_i$ ($u \in R_i$). If $u \in C_i$, then $u \in C_{i+1}$ up to $C_j$.
        Thus, the index set $\{p \mid u \in C_p\}$ is a contiguous interval, and the bags containing $u$ induce a subpath.
        
        \item \textbf{Vertices in $H$.}
        Let $u \in V(H)$. Since $H$ is a plane subgraph consisting of caterpillars, the edges of $H$ incident to $u$ form a contiguous subsequence $e_\ell, e_{\ell+1}, \ldots, e_r$.
        Consequently, $u$ is explicitly contained in every bag from $B^L_\ell$ to $B^R_r$ (the \emph{spine interval}).
        
        Note that $u$ cannot be in $C_i$ for any $\ell \leq i \leq r$, because the sets $L_i$ and $R_i$ exclude the endpoints of $e_i$. However, $u$ may belong to covers outside this range.
        If $u \in C_p$ for some $p > r$, then $u$ lies to the left of $e_p$. By \Cref{lem:Koenig-cover}~(II), $u$ must also be in $C_{r+1}, \ldots, C_{p-1}$. Since $u \in C_{r+1}$, it is included in the bag $B^R_r$ (which contains $C_{r+1}$). This connects the spine interval to the right.
        Symmetrically, if $u \in C_p$ for some $p < \ell$, monotonicity implies $u \in C_{p+1}, \ldots, C_{\ell-1}$. Since $u \in C_{\ell-1}$, it is included in $B^L_\ell$, connecting the spine interval to the left. Therefore, the union of bags containing $u$ forms a subpath. \qedhere
    \end{itemize}
\end{description}
\end{proof}

\begin{figure}[t]
  \centering
  \includegraphics[]{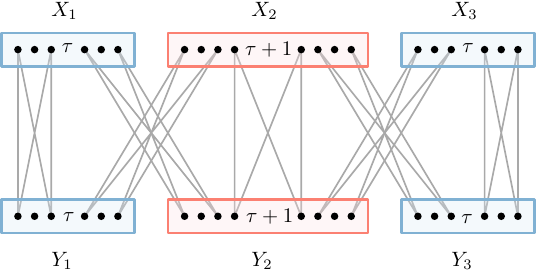}
  \caption{Illustration of the construction of the graph $G_k$.}
  \label{fig:lowerbound}
\end{figure}

\subparagraph*{A Lower Bound on the Pathwidth.} We construct a graph with a $2$-layer $k$-matching-planar drawing that has a $K_{3\lfloor k/2\rfloor+2}$ minor.
Let $\tau=\lfloor\frac{k}{2}\rfloor$.
Define a bipartite graph $G_k = (X \cup Y, E)$ with bipartition $X$ and $Y$, where each side is partitioned into three disjoint blocks:
\[
X = X_{1} \,\dot\cup\, X_{2} \,\dot\cup\, X_{3},
\qquad
Y = Y_{1} \,\dot\cup\, Y_{2} \,\dot\cup\, Y_{3},
\]
and the block sizes are
$\lvert X_{1}\rvert = \lvert X_{3}\rvert = \tau$, $\lvert X_{2}\rvert  = \tau+1$, $\lvert Y_{1}\rvert = \lvert Y_{3}\rvert = \tau$, and $\lvert Y_{2}\rvert = \tau+1$.

The edge set $E$ consists of all edges between the following
pairs of blocks:
\[
E = E(X_{1},Y_{1}) \,\cup\, 
    E(X_{1},Y_{2}) \,\cup\,
    E(X_{2},Y_{1}) \,\cup\,
    E(X_{2},Y_{2}) \,\cup\,
    E(X_{2},Y_{3}) \,\cup\,
    E(X_{3},Y_{2}) \,\cup\,
    E(X_{3},Y_{3}),
\]
where for any two vertex sets $U$ and $V$, $E(U,V) := \{\, uv : u\in U,\ v\in V \,\}$ (see \cref{fig:lowerbound}).

\begin{restatable}{theorem}{pathwidthlower}
\label{thm:pathwidth-lower}%{\let\thefootnote\relax\footnotetext{\hspace{-3.1mm}$^\star$ \hspace{0.4mm}The proofs of the results marked with $\star$ are only given in the full version due to space restrictions.}}
    For any $k$ the graph $G_k$ is $2$-layer $k$-matching-planar and $\operatorname{pw}(G_k) = 3\lfloor \tfrac{k}{2}\rfloor+1$.
\end{restatable}
\begin{proof}

First, we describe a $2$-layer drawing of $G_k$ and show that it is $k$-matching-planar. Place the vertices of $X$ on a horizontal line $\ell_X$ and the vertices of $Y$ on a parallel horizontal line $\ell_Y$ below $\ell_X$.
Along $\ell_X$, order the vertices from left to right so that all vertices of $X_1$ come first, followed by all vertices of $X_2$, and finally all vertices of $X_3$. Within each block the order is arbitrary.
Similarly, along $\ell_Y$, order the vertices from left to right so that all vertices of $Y_1$ come first, followed by all vertices of $Y_2$, and finally all vertices of $Y_3$, again with arbitrary order inside each block. Draw every edge of $G_k$ as a straight-line segment between its endpoints. This is a valid $2$-layer drawing of $G_k$.

Let $e=ab$ with $a\in X$ and $b\in Y$. Two edges cross if and only if their endpoints appear in opposite orders on the two layers.
The drawing is symmetric under reflection and under exchanging the two layers.
Hence, up to symmetry, it suffices to consider the following three cases.

\begin{itemize}
    \item \textbf{Case 1: $a\in X_1$, $b\in Y_1$.}
    Every edge crossing $e$ has at least one endpoint in the set $(X_1\cup Y_1)\setminus\{a,b\}$, which has size $2\tau-2\leq k-2$.
    \item \textbf{Case 2: $a\in X_1$, $b\in Y_2$.}
    Every edge crossing $e$ has at least one endpoint in the set $(Y_1\cup Y_2 )\setminus\{b\}$, which has size $2\tau\leq k$.
    \item \textbf{Case 3: $a\in X_2$, $b\in Y_2$.}
    Every edge crossing $e$ has at least one endpoint in the set
$(X_2\cup Y_2)\setminus\{a,b\}$, which has size $2\tau \leq k$.

\end{itemize}

\noindent Thus any matching of crossing edges of $e$ has size at most $k$. Therefore, the described $2$-layer drawing of $G_k$ is $k$-matching-planar.

Note that $\lvert X\rvert=\lvert Y\rvert=3\tau+1=3\lfloor k/2\rfloor+1$.
Let $n=3\tau+1$.
As $G_k$ is a subgraph of the complete bipartite graph $K_{n,n}$ we have $\operatorname{pw}(G_k) \le \operatorname{pw}(K_{n,n}) = n = 3\lfloor k/2\rfloor+1$.

We now show that this is the minimum possible width by
finding a $K_{3\tau+2}$ minor in $G_k$. First we take a maximum matching on the edges $E(X_1,Y_1)$ and contract every vertex of $Y_1$ along its matching edge to the other endpoint in $X_1$. As for every pair $u$, $v$ of vertices in $X_1$ there was a vertex $y$ in $Y_1$ that was contracted to $u$ and had edges to both $u$ and $v$, there now is a direct edge between $u$ and $v$. Therefore, the vertices of $X_1$ now form a complete graph. Note that every vertex in $X_1$ now has an edge to every vertex in $X_2$ and $Y_2$. Similarly, we take maximum matchings on $E(X_2,Y_3)$ and $E(X_3,Y_2)$ and contract the vertices of $X_3$ and $Y_3$ along their matching edges to the corresponding endpoints in $X_2$ and $Y_3$. This results in a complete graph on the vertices of $X_2$ as well as a complete graph on the vertices of $Y_2$. Since between $X_2$ and $Y_2$ there is a complete bipartite graph and $X_1$ has edges to all vertices of $X_2\cup Y_2$ the resulting graph is a complete graph on $3\tau +2$ vertices.
That is $\operatorname{pw}(G_k)\ge 3\tau+1 =  3\lfloor k/2\rfloor+1$.
\end{proof}

\section{NP-hardness of One-Sided $k$-Matching-Planarity}
\label{sec:np-hardness}

In \cite{OneSidedKPlanar} it was proven that the problem of recognizing $2$-layer $k$-planar graphs given an order on the vertices of one side is NP-hard. This was shown by a reduction from \textsc{$k$-way Partition}, which is strongly NP-hard, as it is a special case of \textsc{3-Partition}. Using similar ideas, we show that the problem \textsc{OneSided-$k$-MatchingPlanarity} is NP-hard as well. The problem asks, given a bipartite graph $G=(X\cup Y, E)$ and an order $\prec_X$ on the vertices of $X$ and some $k$, if there is a linear order $\prec_Y$ on $Y$ so that the $2$-layer drawing induced by $\prec_X$ and $\prec_Y$ is $k$-matching-planar.

\begin{center}
\begin{tcolorbox}[width=13cm]
\begin{center}
 \begin{tabular}{rl}
\textbf{Problem:} & \textsc{OneSided-$k$-MatchingPlanarity}\\
 \textbf {Given:} & A bipartite graph $G=(X\cup Y, E)$, a linear order $\prec_X$ on $X$ and an\\ & integer $k\geq 0$\\
 \textbf {Question:} & Is there a linear order $\prec_Y$ on $Y$ so that the $2$-layer drawing of $G$\\
 & induced by $\prec_X$ and $\prec_Y$ is $k$-matching-planar? \end{tabular}
\end{center}
\end{tcolorbox}
\end{center}

\begin{theorem}\label{thm:onesided-nphard}
    The problem \textsc{OneSided-$k$-MatchingPlanarity} is NP-complete.
\end{theorem}

\begin{proof}
The containment in NP is easy to see as testing whether a graph is $k$-matching-planar can be done in polynomial time. To show the hardness, we reduce from the problem \textsc{$k$-way Partition}. The construction is similar to the construction used in \cite{OneSidedKPlanar}, though our construction is somewhat more involved and needs some additional structures.

\begin{center}
\begin{tcolorbox}[width=13cm]
\begin{center}
 \begin{tabular}{rl}
\textbf{Problem:} & \textsc{$k$-way Partition}\\
 \textbf {Given:} & A set of $n$ positive integers $S=\{s_1,\dots,s_{n}\}$ and an integer $k\geq 2$\\
 \textbf {Question:} & Is there a partition of $S$ into $k$ sets $S_1,\dots,S_k$ such that,\\
 & for each set the sum of its elements is exactly $T=\frac{1}{k}\sum_{s\in S}s$? \end{tabular}
\end{center}
\end{tcolorbox}
\end{center}

\paragraph*{The construction}
We ask if there is a $k'$-matching-planar drawing by choosing $k' = knT+k-1$.
The bipartite graph $G=(X\cup Y, E)$ is constructed from $\mathcal{S}$ as follows: 
The set $X$ will consist of the following sets of vertices (see \Cref{fig:np_hardness}).
\begin{gather*}
    B^X=\{b^X_1,\dots, b^X_{k'+1}\}, \quad B'^X=\{b'^X_1,\dots, b'^X_{k'+1}\},\\
    Q^X_1=\{q^X_{1,1},\dots, q^X_{1,nT+1}\}, \dots, Q^X_{k-1}=\{q^X_{k-1,1},\dots, q^X_{k-1,nT+1}\},\\
    Q'^X_1=\{q'^X_{1,1},\dots, q'^X_{1,nT+1}\}, \dots, Q'^X_{k-1}=\{q'^X_{k-1,1},\dots, q'^X_{k-1,nT+1}\},\\
    A^X_1=\{a^X_{1,1},\dots,a^X_{1,ns_1}\}, \dots, A^X_{n}=\{a^X_{n,1},\dots,a^X_{n,ns_{n}}\},\\
    \{p_1,\dots,p_{k-1}\}, \quad \{p''_1,\dots,p''_{k-1}\}
\end{gather*}
By $A^X$ we denote $\bigcup_{1\leq i\leq n}A^X_i$, by $Q^X$ we denote $\bigcup_{1\leq i\leq k-1}Q^X_i$ and by $Q'^X$ we denote $\bigcup_{1\leq i\leq k-1}Q'^X_i$.
The set of vertices $Y$ for the other side consists of the following vertex sets: 
\begin{gather*}
    B^Y=\{b^Y_1,\dots, b^Y_{k'+1}\}, \quad B'^Y=\{b'^Y_1,\dots, b'^Y_{k'+1}\},\\
    Q^Y_1=\{q^Y_{1,1},\dots, q^Y_{1,nT+1}\}, \dots, Q^Y_{k-1}=\{q^Y_{k-1,1},\dots, q^Y_{k-1,nT+1}\},\\
    Q'^Y_1=\{q'^Y_{1,1},\dots, q'^Y_{1,nT+1}\}, \dots, Q'^Y_{k-1}=\{q'^Y_{k-1,1},\dots, q'^Y_{k-1,nT+1}\},\\
    A^Y_1=\{a^Y_{1,1},\dots,a^Y_{1,knT}\}, \dots, A^Y_{n}=\{a^Y_{n,1},\dots,a^Y_{n,knT}\},\\
    \{p'_1,\dots,p'_{k-1}\}
\end{gather*}

Again, we use $A^Y$ and $Q^Y$ and $Q'^Y$ to denote the union of the corresponding sets.
The sets $B^X$ and $B^Y$ form a complete bipartite graph, as well as the sets $B'^X$ and $B'^Y$. For every $i$, with $1\leq i<k$, the sets $Q_i^X$ and $Q_i^Y$ as well as $Q_i'^X$ and $Q_i'^Y$ also form a complete bipartite graph. For every $i$ and $j$, with $1\leq i<k$ and $1\leq j\leq nT+1$, an edge is placed between $q_{i,j}^Y$ and $b^X_{(i-1)\cdot (nT+1)+j}$ as well as between $q_{i,j}'^Y$ and $b'^X_{i\cdot (nT+1)+j}$. Also for every $i$, where $1\leq i \leq n$, the sets $A^X_i$ and $A^Y_i$ form complete bipartite graphs. Finally, for every $i$, where $1\leq i <k$, there is an edge between $p_i$ and $p'_i$ as well as between $p'_i$ and $p''_i$. The linear order $\prec_X$ on $X$ is defined as follows:
\begin{align*}
    B^X&\prec_X p_1\prec_XQ_1^X\prec_X p_2\prec_X\dots\prec_XQ_{k-1}^X\prec_XA_1^X\prec_X\dots\prec_XA^X_{n}\prec_XQ_1'^X\prec_X p_1''\prec_X\\
    &\dots\prec_XQ_{k-1}'^X\prec_Xp''_{k-1}\prec_XB'^X
\end{align*}
where the order on sets $B^X,Q_1^X,\dots,A^X_1,\dots$ is defined according to the indices of their respective elements.

\begin{figure}[htbp]
  \centering
  \includegraphics{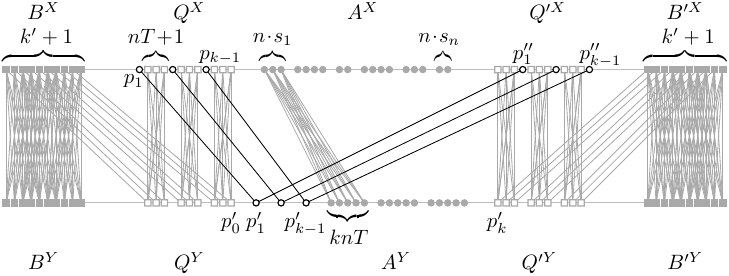}
  \caption{Illustration of the construction for the hardness proof.}
  \label{fig:np_hardness}
\end{figure}

\paragraph*{Soundness and Completeness}

First we argue that if there is a partition into $k$ sets $S_1,\dots, S_k$ such that for each $S_i$ the sum of its elements is $T$ then there is a linear order $\prec_Y$ on $Y$ such that the induced $2$-layer drawing of $G$ is $k'$-matching-planar. The linear order $\prec_Y$ will have to satisfy the following orderings:
\[B^Y\prec_Y Q_1^Y\prec_Y\dots\prec_YQ^Y_{k-1}\prec_Yp'_1\prec_Y\dots\prec_Yp'_{k-1}\prec_YQ'^Y_1\prec_Y\dots\prec_YQ'^Y_{k-1}\]
The vertices of $A^Y$ are sorted into $\prec_Y$ between $Q^Y_{k-1}$ and $Q'^Y_1$ according to $S_1,\dots, S_k$. For each $S_i=\{s_{i_{1}}, \dots, s_{i_{\ell}}\}$ we place the vertices of $A^Y_{{i_1}},\dots ,A^Y_{{i_\ell}}$ between $p'_{i-1}$ and $p'_{i}$ where $p'_0$ is the rightmost vertex of $Q^Y_{k-1}$ in $\prec_Y$ and $p'_k$ is the leftmost vertex of $Q'^Y_{1}$ in $\prec_Y$. We claim that the drawing induced by $\prec_X$ and $\prec_Y$ is $k'$-matching-planar. This is proven by considering  each crossed edge and distinguishing cases according to the location of its endpoints. 
\begin{itemize}
    \item $B^X$ and $B^Y$ respectively $B'^X$ and $B'^Y$: For each edge $e$ between $B^X$ and $B^Y$ all edges crossing $e$ have an endpoint in $B^X$. Hence, there are at most $\lvert B^X\rvert -1 = k'$ edges in any matching crossing $e$. The same conclusion holds for $B'^X$ and $B'^Y$.
    \item $B^X$ and $Q^Y$ respectively $B'^X$ and $Q'^Y$: Consider an edge $e=b q^Y_{i,j}$, where $b=b^X_{(i-1)\cdot (nT+1)+j}$. Such an edge is only crossed by edges having endpoints in $B^X$ that are to the right of $b$ and edges with endpoints in $Q^Y$ that are to the left of $q^Y_{i,j}$. This means that in total there are at most $k'+1-(i\cdot (nT+1)-(nT+1)+j)+(i-1)\cdot(nT+1)+j-1=k'$ edges in any matching crossing $e$. The same conclusion holds for edges between vertices of $B'^X$ and $Q'^Y$.
    \item $Q^X$ and $Q^Y$ respectively $Q'^X$ and $Q'^Y$: Let $e$ be an edge between vertices of $Q^X_i$ and $Q^Y_i$. The edge is crossed by the edges $p_jp'_j$, for any $j$ with $1\leq j\leq i$, as well as edges having endpoints in the other $nT$ many vertices in $Q^Y_i$. Additionally, there are $(k-1-i)(nT+1)$ edges between $Q_j^Y$ and $B^X$, for any $j$ with $i<j\leq k-1$, crossing $e$. This means that any matching crossing $e$ has size at most $i+nT+(k-1-i)(nT+1)=knT+k-1-i(nT) \leq  k'$. 
    The same  conclusion holds for edges between vertices of $Q'^X$ and $Q'^Y$.
    \item $p_ip'_i$ and $p'_ip''_i$: Suppose that $e=p_ip'_i$.
    The edges between $A_j^X$ and $A_j^Y$ cross~$e$ only if $s_j\in S_t$ for some $t\leq i$.
    So, the edges between $A^X$ and $A^Y$ that cross $e$ have $i(nT)$ common endpoints in $A^X$, as the sum of elements of each $S_t$ is $T$.
    An edge with an endpoint in $Q^X$ crosses $e$ only if it is between $Q^X_j$ and $Q^Y_j$ for some $j\geq i$.
    These edges have $(k-i)(nT+1)$ endpoints in $Q^X$.
    Finally, there are $i-1$ edges $p_jp_j'$ for $j< i$ crossing $e$. The size of any matching crossing $e$ is therefore at most $i(nT)+(k-i)(nT+1)+i-1=knT+k-1=k'$.
    The same conclusion holds for the edges  $p'_ip''_i$.
    \item $A^X$ and $A^Y$: Consider an edge $e$ between $A^X_i$ and $A^Y_i$. It is crossed by $k-1$ edges with some endpoint $p'_j$, for $1\leq j\leq k-1$. There is also a total of $knT-1$ vertices in $A^X$ that may be endpoints of edges crossing $e$. Therefore, there are at most $knT+k-1$ edges in any matching crossing $e$.
\end{itemize}

As each edge is crossed by a maximum matching of size at most $k'$ the $2$-layer drawing of $G$ induced by $\prec_X$ and $\prec_Y$ is $k'$-matching-planar.

Now suppose that there is a linear order $\prec_Y$ on $Y$ such that $\prec_Y$ and $\prec_X$ induce a $2$-layer $k'$-matching-planar drawing on $G$. To extract a $k$-way partition from $\prec_Y$ we first need to prove some structural properties that $\prec_Y$ must satisfy.

\begin{restatable}{lemma}{qordered}
\label{lem:q-ordered}
    For each $p_i'$ where $1 \leq i \leq k-1$, each $q_Y\in Q^Y$, and each $q_Y'\in Q'^Y$ it holds that $q_Y \prec_Y p_i'\prec_Y q_Y'$.
\end{restatable}
\begin{proof}
    Assume that there is some $p'_i$ for which there either is a $q_Y\in Q^Y$ for which $p_i'\prec_Y q_Y$ or there is a $q_Y'\in Q'^Y$ for which $q_Y'\prec_Y p_i'$. As the two cases are symmetric it suffices to assume that there is a $q_Y\in Q^Y$ for which $p_i'\prec_Y q_Y$. Since $p_i'\prec_Y q_Y$ there is an edge between $B^X$ and $q_Y$ that crosses $p_ip_i'$ and there are edges between $Q^X$ and $q_y$ that cross $p_i'p_i''$.
    %and there are edges between $q_Y$ and points from $Q'_i^Y$ that cross $p'_ip''_i$.
    There are also at least $k(nT+1)-1$ other vertices in $Q^Y\cup Q'^Y$ that are endpoints of edges that cross either $p_ip_i'$ or $p_i'p_i''$ and there are $knT$ vertices in $A^X$ that have edges that cross either $p_ip_i'$ or $p_i'p_i''$. Also, for each $p_\ell'$ where $\ell\neq i$ there is one edge adjacent to $p_\ell'$ that must cross either $p_ip_i'$ or $p_i'p_i''$. In total the size of maximum matchings for edges $p_ip_i'$ and $p_i'p_i''$ combined are $k(nT+1)-1+knT+(k-2)+2=2knT+2(k-1)+1$ and therefore one of them is crossed by a matching of size at least $k'+1$.
\end{proof}

\begin{restatable}{lemma}{pordered}
\label{lem:p-ordered}
    For every $p_i'$ and $p_j'$ where $i< j$ it must hold that $p_i'\prec_Y p_j'$.
\end{restatable}
\begin{proof}
This can be proven essentially the same way as in \cite{OneSidedKPlanar}. Suppose there are $p_i'$ and $p_j'$ where $i<j$ and $p_j' \prec p_i'$. Then $p_i'p_i$ and $p_j'p_j$ cross as well as $p_i'p_i''$ and $p_j'p_j''$. There is a total of $knT$ vertices in $A^X$ that have edges that either cross $p_ip_i'$ or $p_i'p_i''$. There is also a total of $k(nT+1)$ vertices in $Q^X\cup Q'^X$ that have edges that need to cross either $p_ip_i'$ or $p_i'p_i''$. Also, for each $p_\ell'$ where $\ell\neq i$ and $\ell\neq j$ there is one edge adjacent to $p_\ell'$ that must cross either $p_ip_i'$ or $p_i'p_i''$. In total the size of maximum matchings for edges $p_ip_i'$ and $p_i'p_i''$ combined are $2knT+k+(k-3)+2=2knT+2(k-1)+1$ and therefore one of them must have a maximum matching of size at least $k'+1$.
\end{proof}

\begin{restatable}{lemma}{anosplit}
\label{lem:a-no-split}
    If according to $\prec_Y$ there is some $a_Y\in A^Y_j$ between $p'_{i}$ and $p'_{i+1}$ where $1\leq i \leq k-2$ then every vertex of $A^Y_j$ must be between $p'_{i}$ and $p'_{i+1}$. The same is true analogously for $a_Y\in A^Y_j$ left of $p'_{1}$ and $a_Y\in A^Y_j$ to the right of $p'_{k-1}$. 
\end{restatable}
\begin{proof}

Suppose there is some $a_Y\in A^Y_j$ between $p'_{i-1}$ and $p'_{i}$ and some $a_Y'\in A^Y_j$ between $p'_{\ell}$ and $p'_{\ell+1}$ where $\ell> i$. There is a total of $k(nT+1)$ vertices in $Q^X\cup Q'^X$ that have edges that need to cross either $p_ip_i'$ or $p_i'p_i''$. There is also a total of $k-2$ edges $p_jp'_j$ respectively $p'_jp_j''$ for $j\neq i$ crossing either $p_ip'_i$ or $p'_ip''_i$. And finally there is a total of $knT$ vertices in $A^X$ that have edges crossing either $p_ip'_i$ or $p'_ip''_i$. One of these vertices will have an edge that crosses $p_ip'_i$ as well as an edge that crosses $p'_ip''_i$. In total the size of maximum matchings for edges $p_ip_i'$ and $p_i'p_i''$ combined are $k(nT+1)+knT+1+(k-2)=knT+k+knT+1+(k-2)=2knT+2(k-1)+1$ and therefore one of them must have a maximum matching of size at least $k'+1$.
\end{proof}

By Lemmas \ref{lem:q-ordered} and \ref{lem:p-ordered} we know that $\prec_Y$ must satisfy the following inequality:
\[Q^Y\prec_Yp'_1\prec_Y\dots \prec_Y p'_{k-1}\prec_Y Q'^Y\]
It is not hard to see that \textsc{$k$-way Partition} is strongly NP-hard as it is a special case of \textsc{3-Partition}, where for every element $s$ it holds that $T/4<s<T/2$, which is known to be strongly NP-hard~\cite{GareyJ79}. 
We claim that 
\[S_i=\{s_j \mid p'_{i-1}\prec_Y a_Y\prec_Y p'_{i},\ a_Y\in A_j^Y\},\]
where $p'_0$ is to the left of all vertices in $A^Y$ and $p'_k$ is to the right of all vertices in $A^Y$, is a partition of $\mathcal{S}$ for which the sum of every set is exactly $T$.
This can be shown similarly as in~\cite{OneSidedKPlanar} by induction on $i$. For $i=1$ we argue that for $p_1p'_1$ there are $(k-1)(nT+1)$ endpoints of edges crossing in $Q^X$. This means that there can be at most $nT$ endpoints in $A^X$ that have edges crossing $p_1p'_1$ and therefore it must hold that $\sum_{s\in S_1}s\leq T$. For the edge $p'_1p''_1$ there are $nT+1$ endpoints in $Q'^X$ that have edges crossing $p'_1p''_1$, as well as $k-2$ edges $p_jp'_j$ for $j\geq 2$. There are also $knT-n\sum_{s\in S_i}s$ endpoints in $A^X$ with edges crossing $p'_1p''_1$. In total this means there are 
\[nT+1+k-2+knT-n\sum_{s\in S_i}s=knT+k-1+nT-n\sum_{s\in S_i}s\leq k'\] vertices that have edges crossing $p'_1p''_1$, which means that $\sum_{s\in S_i}s\geq T$. For $i>1$ we can proceed similarly. By the induction hypothesis we know that there are $(i-1)nT$ endpoints in $A^X$ for edges that cross $p_{i-1}p'_{i-1}$ as well as $p_ip'_i$. There are also $(i-1)$ edges $p'_jp''_j$ for $j<i$ that cross $p_ip'_i$ and $(k-i)(nT+1)$ endpoints in $Q^X$ for edges that cross $p_ip'_i$. This means in total there are 
\[(i-1)nT+i-1+(k-i)(nT+1)+n\sum_{s\in S_i}s=knT+k-1-nT+n\sum_{s\in S_i}s\leq k'\]
endpoints in $X$ for edges crossing $p_ip'_i$ and $\sum_{s\in S_i}s\leq T$. For the edge $p'_ip''_i$ there are in total 
\[i(nT+1)+k-1-i+knT-(i-1)(nT)-n\sum_{s\in S_i}s=knT+k-1+nT-n\sum_{s\in S_i}s \leq k'\]
distinct endpoints in $X$ that have edges crossing $p'_ip''_i$, which means that $\sum_{s\in S_i}s\geq T$.

Hence, $\sum_{s\in S_i}s=T$ holds for $1\leq i \leq k-1$. For $i=k$ we can simply argue that there are exactly $nT$ vertices left in $A^X$ that have edges with endpoints to the right of $p'_{k-1}$. This concludes the proof of the hardness of \textsc{OneSided-$k$-MatchingPlanarity}.
\end{proof}

Using their reduction, the authors of \cite{OneSidedKPlanar} could show that under ETH there is no $2^{o(\vert Y\vert)}\text{poly}(\vert \mathcal{I}\vert)$-time algorithm for recognizing \textsc{OneSided $k$-Planarity}, where $\vert \mathcal{I}\vert$ denotes the size of the instance. This was a direct consequence of their construction.

For the reduction above we do not immediately obtain such a result. The main obstacle is that in the construction above the size of $Y$ is not bounded in terms of $n$.

By slightly tweaking the construction above, the reduction can be modified to a reduction from \textsc{3-Partition}. 

\section{FPT-algorithm for One-Sided $k$-Matching-Planarity}
\label{sec:fpt}

In this section we describe an FPT-algorithm to decide \textsc{OneSided-$k$-MatchingPlanarity}.
Our algorithm is similar to the algorithms presented in~\cite{KOW25} for recognizing $2$-layer $k$-planar graphs.
However, the adaptation to $k$-matching-planarity is non-trivial and needs different structural insights presented next.

\begin{restatable}{lemma}{redundant}\label{lem:redundant}
	Let $k$ be an integer, $G=(X\cup Y, E)$ be a bipartite graph, $u\in X\cup Y$, and $s=\lvert N(u)\rvert$.
	If there are at least $2s+1$ vertices $v\in X\cup Y$ with $N(u)=N(v)$, then $G$ has a $2$-layer $k$-matching-planar drawing if and only if $G-u$ has such a drawing.
\end{restatable}
\begin{proof}
	The statement is clear if $s=0$, so consider $s\geq 1$.
	If $G$ has a $2$-layer $k$-matching-planar drawing then clearly $G-u$ has such a drawing.
	So assume that $G-u$ has a $2$-layer $k$-matching-planar drawing $D'$.
	By assumption, there are at least $2s$ vertices $v_1,\ldots,v_{2s}$ in $G-u$ with $N(v_i)=N(u)$ for each $i$.
	Note that the neighborhood does not change from $G$ to $G-u$ as $G$ is bipartite.
	Assume that $v_1,\ldots,v_{2s}$ appear in this order in $D'$.
	We claim that inserting $u$ directly to the right of $v_s$ yields a $k$-matching-planar drawing $D$ of $G$.
	To see this we will find, for each edge $e$ in $G$, a vertex cover of size at most $k$ for the edges crossing $e$ in~$D$.
    
	First consider an edge $e=xy$ not incident to $u$.
	Then there is a vertex cover $S'$ of size at most $k$ for the edges crossing $e$ in $D'$.
	Without loss of generality, we assume that $x\in X$ and $y$, $u\in Y$ with $u\prec_Y y$.
	Let $X'=\{x'\in N(u)\colon x \prec_X x'\}$ denote the set of neighbors $x'$ of $u$ such that the edge $x'u$ crosses $e$ in $D$.
    Then $\lvert X'\rvert \leq s$.
	By the assumption $u\prec_Y y$, the set $S=(S'\cup X')\setminus \{v_1,\ldots,v_s\}$ is a vertex cover of the edges crossing $e$ in $D$.	   
	Moreover, all edges of the complete bipartite graph between $\{v_1,\ldots,v_s\}$ and $X'$ cross $e$ in $D'$ and, hence, either all vertices from $X'$ or all vertices from $\{v_1,\ldots,v_s\}$ are in $S'$.
    This shows $\lvert S\rvert \leq \lvert S'\rvert \leq  k$.
    \cref{fig:reduce-vertex-copies} shows illustrations of both cases.

    \begin{figure}[t]
        \centering
        \includegraphics{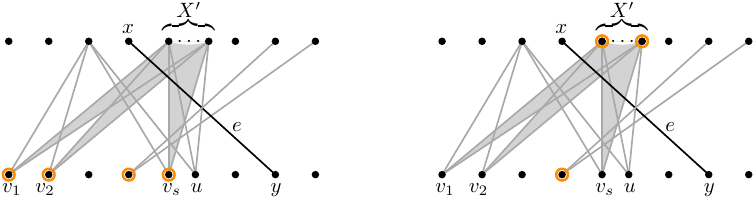}
        \caption{The edge $e$ is crossed by a complete bipartite graph between $\{v_1,\ldots,v_s\}$ and $X'$. All vertices from $X'$ or all vertices from $\{v_1,\ldots,v_s\}$ are in a vertex cover $S'$ (orange) of these edges.}
        \label{fig:reduce-vertex-copies}
    \end{figure}
	
	Now consider an edge $e=xu$ incident to $u$.
	There is a vertex cover $S'$ of size at most $k$ for the edges crossing the edge $xv_s$ in $D'$. 
	Let $X'=\{x'\in N(u)\colon x \prec_X x'\}$ denote the set of neighbors $x'$ of $u$ such that each edge $xv_i$, $i\leq s$, crosses $e$ in $D$.
    Then $\lvert X'\rvert \leq s-1$ (as $x\not \in X'$).
	Similar as above, $S=(S'\cup X')\setminus \{v_1,\ldots,v_{s-1}\}$ is a vertex cover of the edges crossing $e$ in $D$ and $\lvert S\rvert \leq \lvert S'\rvert \leq k$.
\end{proof}

We call a graph \emph{nonredundant} if for each vertex $u$ there are at most $2\lvert N(u)\rvert$ other vertices $v$ with $N(v)=N(u)$.
Lemma~\ref{lem:redundant} states that we can iteratively remove vertices from a graph $G$ until it is nonredundant, so that the resulting graph is $2$-layer $k$-matching-planar if and only if $G$ is.

\begin{restatable}{lemma}{windowsize}\label{lem:windowsize}
	Let $D$ be a $2$-layer $k$-matching-planar drawing of a nonredundant graph $G=(X\cup Y, E)$.
    For each pair of edges $xy$ and $x'y'$, with $x$, $x'\in X$, $x \prec_X x'$, and $y' \preceq_Y y$ (possibly $y=y'$), there are at most $2k + (4k+2) 2^{2k+1}$ vertices between $x$ and $x'$ in $D$.
\end{restatable}
\begin{proof}
    Let $X'\subseteq X$ denote the set of vertices between $x$ and $x'$ in $D$.
    Since $D$ is $k$-matching-planar, there is a set $S$ of at most $2k+1$ vertices such that every edge crossing $xy$ or $x'y'$ has an endpoint in $S$.
    As $G$ is nonredundant, each vertex in $X'$ is incident to at least one edge.
	Each such edge $ab$, with $a\in X'$, is crossing at least one of $xy$ or $x'y'$ or, in case $y=y'$ may also end at $y$.
	Hence, either $a\in S$ or $N(a)\subseteq S\cup\{y\}$.
	For each subset $S'\subseteq S\cup\{y\}$ there are, by \cref{lem:redundant}, at most $2\lvert S'\rvert \leq 4k+2$ vertices $a\in X'$ with $N(a)=S'$ as $G$ is nonredundant.
	This shows that $\lvert X'\rvert \leq 2k + (4k+2) 2^{2k+1}$.
\end{proof}

\begin{corollary}\label{cor:degreeBound}
	Each vertex of a nonredundant $2$-layer $k$-matching-planar graph $G$ has degree at most $2 + 2k + (4k+2) 2^{2k+1}$.
\end{corollary}

\begin{figure}
    \centering
    \includegraphics{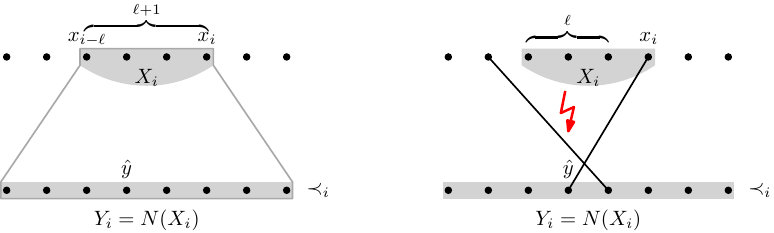}
    \caption{A window consists of $\ell+1$ consecutive vertices from $X$ and their neighbors in $Y$ (left). Edges coming from the left of the window do not cross edges coming from $x_i$ or from even further to the right as the graph is nonredundant (right).}
    \label{fig:algorithm-window}
\end{figure}

Next, we describe the algorithm for a nonredundant bipartite graph $G(X\cup Y, E)$ with a given linear order $\prec_X$ of $X$.
Let $\ell=1 + 2k + (4k+2) 2^{2k+1}$ and let $x_1, \ldots, x_{\lvert X\rvert}$ denote the vertices in $X$ in the given order $\prec_X$.
For each $i$, with $1 \leq i \leq \lvert X\rvert$, let $X_{\leq i}=\{x_1,\ldots,x_i\}$ and $Y_{\leq i}=N(X_{\leq i})$.
Further, for $\ell+1 \leq i \leq \lvert X\rvert$, let $X_i=\{x_{i-\ell},\ldots,x_i\}$ and $Y_i=N(X_i)$.
The sets $X_i$ and $Y_i$ define a \emph{window of size $\ell$}, that is the graph $G[X_i\cup  Y_i]$.
For any $2$-layer $k$-matching-planar drawing the corresponding linear order $\prec_Y$ of $Y$ satisfies the following property due to Lemma~\ref{lem:windowsize} and the choice of $\ell$ (see \cref{fig:algorithm-window}).

\begin{observation}\label{fact:windowseparate}
    There is $\hat{y} \in Y_i$ such that for any $j\leq i-1-\ell$ and $j'\geq i$ we have $N(x_j) \prec_Y \hat{y} \preceq_Y N(x_{j'})$, that is, $\hat{y}$ separates the neighbors coming from the left of the window from the neighbors coming from the right of the window.
\end{observation}

We will check recursively whether any $2$-layer $k$-matching-planar drawing of a window extends to the graph $G[X_{\leq i}, Y_{\leq i}]$ to its left using the following \emph{merge operation} for linear orders:
Given a linear order $\tau$ of some set $X$ and a linear order $\sigma$ of some set $Y$ that agree on $X \cap Y$, let $\tau \triangleleft \sigma$ denote the unique linear order of $X\cup Y$ that extends both $\tau$ and $\sigma$ and places vertices from $X$ to the left of vertices in $Y$ whenever there is a choice.
An example is given in \cref{fig:algorithm-merge}.

Consider any integer $i$, with $\ell+1 \leq i \leq \lvert X\rvert$, any linear order $\prec_i$ of $Y_i$, and any $\hat{y}\in Y_i$.
Here $\prec_i$ specifies a drawing of the window $G[X_i\cup  Y_i]$ and $\hat{y}$ stands for the position in the window which separates insertions of further vertices from the left and from the right of the window (described in \cref{fact:windowseparate}).
We define the boolean predicate

\begin{description}
    \item[$\phi(i, \prec_i, \hat{y}) = \text{true}$] if and only if
    
    there is a $2$-layer $k$-matching-planar drawing of $G[X_{\leq i}\cup Y_{\leq i}]$ such that the induced linear order of $Y_i$ agrees with $\prec_i$ and each vertex from $Y_{\leq i-\ell-1}$ is placed to the left of $\hat{y}$.
\end{description}

\begin{figure}[htbp]
    \centering
    \includegraphics{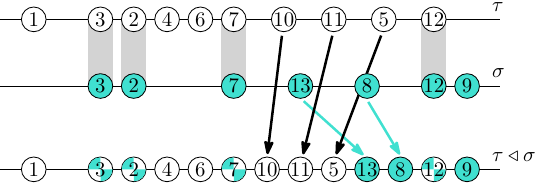}
    \caption{In the merged linear order $\tau\triangleleft\sigma$, vertices from $\tau$ are placed as far left as possible.}
    \label{fig:algorithm-merge}
\end{figure}

Then $G$ has a $2$-layer $k$-matching-planar drawing if and only if $\phi(\lvert X\rvert, \prec_{\lvert X\rvert}, \hat{y})$ is true for some order $\prec_{\lvert X\rvert}$ and some $\hat{y}$.
In the base case $i=\ell+1$, we directly check whether the $2$-layer drawing of $G[X_{\leq i}\cup Y_{\leq i}] = G[X_i,Y_i]$ given by $\prec_i$ is $k$-matching-planar.
For $i>\ell+1$ we will determine $\phi$ recursively by dynamic programming.
We say that a linear order $\prec_{i-1}$ of $Y_{i-1}$ is \emph{compatible} with $\prec_i$, if the orders agree on $Y_i \cap Y_{i-1}$ and, in the merged order $\prec_{i-1} \triangleleft \prec_i$, all neighbors of $x_{i-\ell-1}$ (the only vertex in $X_{i-1}\setminus X_i$) are to the left of $\hat{y}$.
The following lemma shows how $\phi(i,\prec_i,\hat{y})$ is determined recursively.

\begin{restatable}{lemma}{lemrecursion}\label{lem:recursion}
   Let $i > \ell+1$, let $\prec_i$ be a linear order of $Y_i$ such that the corresponding $2$-layer drawing of $G[X_i, Y_i]$ is $k$-matching-planar, and let $\hat{y}\in Y_i$ such that $\hat{y} \preceq_i N(x_i)$.
   Then $\phi(i,\prec_i, \hat{y})$ is true if and only if $\phi(i-1,\prec_{i-1}, \hat{y}')$ is true for some linear order $\prec_{i-1}$ of $Y_{i-1}$ that is compatible with $\prec_i$ and $\hat{y}'$ is the leftmost vertex from $N(x_{i-1}) \cup \{u\in Y_{i-1}\colon \hat{y}  \triangleleft_i u\}$, where $\triangleleft_i$ is the order $\prec_{i-1} \triangleleft \prec_i$.
\end{restatable}
\begin{proof}
	First assume that $\phi(i,\prec_i,\hat{y})$ is true.
    Consider the corresponding $2$-layer $k$-matching-planar drawing $D$ of $G[X_{\leq i}\cup Y_{\leq i}]$ such that the induced linear order of $Y_i$ agrees with $\prec_i$ and each vertex from $Y_{\leq i-\ell-1}$ is placed to the left of $\hat{y}$.
    Let $\prec_{i-1}$ be the linear order of $Y_{i-1}$ obtained from $D$ and let $\hat{y}'$ be defined as in the lemma.
    Clearly $\prec_{i-1}$ and $\prec_{i}$ are compatible.
    We will show that $\phi(i-1, \prec_{i-1}, \hat{y}')$ is true.
    Clearly, the drawing $D$ can be restricted to a $2$-layer $k$-matching-planar drawing of $G[X_{\leq i-1}\cup Y_{\leq i-1}]$ such that the induced linear order of $Y_{i-1}$ agrees with $\prec_{i-1}$.
    Each vertex from $Y_{\leq i-\ell-2}$ is to the left of $\hat{y}$ by assumption and to the left of all vertices from $N(x_{i-1})$ by Lemma \ref{lem:windowsize}.
    Hence, each vertex from $Y_{\leq i-\ell-2}$ is placed to the left of $\hat{y}'$, because $\hat{y}'$ is the leftmost vertex in $N(x_{i-1})$ or we have $\hat{y} \triangleleft_i \hat{y}'$.
    This shows that $\phi(i-1, \prec_{i-1}, \hat{y}')$ is true.
    
    Now assume that $\phi(i-1,\prec_{i-1}, \hat{y}')$ is true for some order $\prec_{i-1}$ of $Y_{i-1}$ that is compatible with $\prec_i$ and $\hat{y}'$ as defined in the lemma.
    So there is a $2$-layer $k$-matching-planar drawing $D'$ of $G[X_{\leq i-1}, Y_{\leq i-1}]$ such that the induced linear order of $Y_{i-1}$ agrees with $\prec_{i-1}$ and each vertex from $Y_{\leq i-\ell-2}$ is placed to the left of $\hat{y}'$.
	Let $\prec'$ be the order of $Y_{\leq i-1}$ in $D'$ (which agrees with $\prec_{i-1}$ on $Y_{i-1}$).
    We first show that $\prec'$ and $\prec_i$ agree on their shared vertices, which allows to merge these orders.
    Then we show that the merged order $\prec' \triangleleft \prec_i$ yields a $2$-layer $k$-matching-planar drawing of $G[X_{\leq i}, Y_{\leq i}]$ and that each vertex from $Y_{\leq i}\setminus Y_i$ is placed to the left of $\hat{y}$ in the merged order. 
    This proves that $\phi(i,\prec_i, \hat{y})$ is true.
    
    The vertices that appear in both orders are given by $Y_i\cap Y_{\leq i-1} = Y_{i}\cap Y_{i-1}$, where the equality holds by means of \cref{fact:windowseparate}.
    Since $\prec_{i-1}$ and $\prec_i$ are compatible, they agree on $Y_i\cap Y_{i-1}$, and hence  $\prec'$ and $\prec_i$ agree on their shared vertices.
    This allows to consider the merged order $\prec' \triangleleft \prec_i$ which we denote by $\prec$.
    Note that $\prec$ and $\prec_{i-1}\triangleleft \prec_i$ agree on $Y_i \cup Y_{i-1}$.
    We next consider the separation point $\hat{y}$.
    Observe that $Y_{\leq i-\ell-1} = Y_{\leq i-\ell-2} \cup N(x_{i-\ell-1})$.
    The vertices in  $N(x_{i-\ell-1})$ are placed to the left of $\hat{y}$ in $\prec$, since $\prec_{i-1}$ is compatible with $\prec_i$.
    The vertices in $Y_{\leq i-\ell-2}$ are placed to the left of $\hat{y}'$ in $\prec_{i-1}$ and hence in $\prec$.
    So if $\hat{y}' \preceq \hat{y}$, then they are clearly left of $\hat{y}$ as well.
    If $\hat{y} \prec \hat{y}'$, then $\hat{y}'$ is the leftmost vertex from $\{u\in Y_{i-1}\colon \hat{y}  \prec u\}$ (and $\hat{y}\not\in Y_{\leq i-1}$).
    Consider some $y\in Y_{\leq i-\ell-2}$.
    If there is $u\in Y_{i-1}$ such that  $y \prec' u \prec' \hat{y}'$ then  $y \prec u \prec \hat{y} \prec \hat{y}'$.
    Otherwise, there is a choice on the order of $y$ and $\hat{y}$ when merging $\prec'$ and $\prec_i$.
    In this case the merge operation places $y$ to the left of $\hat{y}$ by definition.
    Hence, in each case $\hat{y}$ is the desired separation point and all vertices from $Y_{\leq i-\ell-1}$ are placed left of $\hat{y}$ in $\prec$.

    It remains to show that $\prec$ yields a $k$-matching-planar drawing $D$ of $G[X_{\leq i}, Y_{\leq i}]$.
    Restricting $D$ to $G[X_{\leq i-1}, Y_{\leq i-1}]$ yields the $k$-matching-planar drawing $D'$.
    Restricting $D$ to $G[X_{i}, Y_{i}]$ yields a $k$-matching-planar drawing $D_i$ by assumption on $\prec_i$.
    So it suffices to prove that for each edge $e$ in $D$, the edge $e$ and the edges crossing $e$ are either all in $D'$ or all in $D_i$.
    Consider an edge $xy$, with $x\in X$, that is not in $D'$ itself or crossed by some edge not in $D'$.
    Then either $x = x_i$ or $xy$ is crossed by some edge incident to $x_i$.
    In both cases, $\hat{y} \preceq y$ since $\hat{y} \prec_i N(x_i)$ by assumption.
    We showed above that each vertex from $Y_{\leq i-\ell-1}$ is placed left of $\hat{y}$ in $\prec$.
    This implies that $x\in X_i$ and, hence, $xy$ and all edges crossing $xy$ are in $D_i$.
    This shows that $\prec$ yields a $k$-matching-planar drawing $D$ of $G[X_{\leq i}, Y_{\leq i}]$.
    Altogether, $\phi(i,\prec_i, \hat{y})$ is true.
\end{proof}

\begin{theorem}
	For any $n$-vertex bipartite Graph $G$, \textsc{OneSided-$k$-MatchingPlanarity} can be decided in time $O(n^3 + \lvert X\rvert 2^{2^{12k}})$.
\end{theorem}

\begin{proof}
	The algorithm first removes vertices from $G$ until it is nonredundant.
	To this end, the neighborhood of each vertex is computed and compared with the other vertices' neighborhoods.
	For each detected neighborhood, only $2k$ vertices with that neighborhood are kept.
	Removing vertices changes some neighborhoods, so the process is repeated until the graph is non-redundant.
    Since at least one vertex is removed in each repetition, at most $n$  repetitions are needed, each requiring at most $O(n^2)$ time.

    We use the following rough upper bounds to simplify calculations: $\ell\leq 2^{3k+3}$ and each vertex has degree at most $2^{3k+3}$ by \cref{cor:degreeBound}.
    
    The dynamic programming computes the values of all predicates $\phi(i,\prec_i, \hat{y})$ for each $i=\ell+1,\ldots,\lvert X\rvert$, each linear order $\prec_i$ of $Y_i$, and each $\hat{y}\in Y_i$.
    By \cref{cor:degreeBound} we have $\lvert X_i \cup Y_i \rvert \leq 1 + \ell + \ell (2 + 2k + (4k+2) 2^{2k+1}) \leq 2^{6k+6}$.
    So the number of predicates to be computed is in $O(\lvert X\rvert\ 2^{6k+6}!)$.
    For each predicate we have to first check whether $\prec_i$ yields a $k$-matching-planar drawing of $G[X_i, Y_i]$.
    This can be achieved in $O(\lvert X_i\cup Y_i\rvert^2) = O(2^{12k+12})$ time.
    If $i=\ell+1$, then $\phi(i,\prec_i, \hat{y})$ is true if and only if this check returned true.
    If $i>\ell+1$, then we determine $\phi(i,\prec_i, \hat{y})$ using Lemma~\ref{lem:recursion}.
    To consider all orders $\prec_{i-1}$ of $Y_{i-1}$ that are compatible with $\prec_i$, it suffices to consider all $O(2^{(6k+6)2^{6k+6}})$ possibilities to insert the vertices from $Y_{i-1}\setminus Y_i \subseteq N(x_{i-\ell-1})$ into the order $\prec_i$.
    Checking for compatibility with $\prec_i$ can be done in $O(\lvert Y_i\rvert) = O(2^{6k+6})$ time each.
    For each compatible order, the separation point $\hat{y}'$ is computed according to the definition in Lemma~\ref{lem:recursion} and the previously computed predicate $\phi(i-1,\prec_{i-1}, \hat{y}')$ is accessed in $O(\lvert Y_i\rvert)$ time.
    Altogether, the dynamic programming needs $O(\lvert X\rvert\ 2^{6k+6}! (2^{12k+12} + 2^{(6k+6)2^{6k+6}} 2^{6k+6})) = O(\lvert X\rvert 2^{2^{12k}})$ time.
\end{proof}

Note that we make no attempts to optimize the running time in terms of $k$.

\section{Inapproximability of Two-Sided $k$-Matching-Planarity}\label{sec:inapprox-twosided}

In this section we consider the optimization variant of \textsc{TwoSided-$k$-MatchingPlanarity}, asking for the smallest $k$ such that a given bipartite graph admits a $2$-layer $k$-matching-planar drawing. We show that this number cannot be approximated within any constant factor in polynomial time unless $\operatorname{P}=\operatorname{NP}$, even on trees and therefore also NP-hardness.

\begin{center}
\begin{tcolorbox}[width=13cm]
\begin{center}
 \begin{tabular}{rl}
\textbf{Problem:} & \textsc{TwoSided-$k$-MatchingPlanarity}\\
 \textbf {Given:} & A bipartite graph $G=(X\cup Y, E)$\\
 \textbf {Question:} & What is the smallest integer $k$ such that there are linear orders\\ & $\prec_X$ and $\prec_Y$ on $X$ and $Y$ respectively, so that the $2$-layer drawing\\
 & of $G$ induced by $\prec_X$ and $\prec_Y$ is $k$-matching-planar? \end{tabular}
\end{center}
\end{tcolorbox}
\end{center}

\begin{restatable}{theorem}{thminapprox}\label{thm:inapprox}
    For every constant $c\geq 1$, there is no polynomial-time $c$-approximation algorithm for the smallest $k$ such that a given bipartite graph admits a $2$-layer $k$-matching-planar drawing unless $\operatorname{P}=\operatorname{NP}$, even when the input graph is restricted to trees.
\end{restatable}
\begin{proof}
To show the inapproximability we reduce from the \textsc{Bandwidth} problem on trees similar to~\cite[Theorem~16]{KOW25}. The \emph{bandwidth} of a graph $H$, denoted $\operatorname{bw}(H)$, is the minimum integer $b$ such that there is a linear order $\sigma\colon V(H)\to\{1,\dots,|V(H)|\}$ with $|\sigma(u)-\sigma(v)|\leq b$ for every edge $uv\in E(H)$.
This problem is NP-hard to approximate within any constant factor on trees~\cite{DFU11}.

\begin{center}
\begin{tcolorbox}[width=13cm]
\begin{center}
 \begin{tabular}{rl}
\textbf{Problem:} & \textsc{Bandwidth}\\
 \textbf {Given:} & A tree $T$\\
 \textbf {Question:} & What is the smallest integer $b$ such that there is a linear order $\sigma$\\
 & of $V(T)$ with $|\sigma(u)-\sigma(v)|\leq b$ for every edge $uv\in E(T)$?
 \end{tabular}
\end{center}
\end{tcolorbox}
\end{center}

\paragraph*{The construction}
Given a tree $T$ with $n$ vertices, we construct a bipartite graph $G_T=(X_T\cup Y_T, E_T)$ as follows. We first subdivide each edge $uv\in E(T)$ once by introducing a new vertex $w_{uv}$, so the edge $uv$ is replaced by the path $u,w_{uv},v$. Then, for each vertex $v\in V(T)$, we attach a new pendant $y_v$ adjacent to $v$. The bipartition is given by
\[
X_T = V(T),\qquad Y_T = \{w_{uv} : uv\in E(T)\}\cup\{y_v : v\in V(T)\}.
\]
Every edge of $G_T$ has one endpoint in $X_T$ and one in $Y_T$, so $G_T$ is bipartite, and the construction can clearly be carried out in polynomial time. Notice that $G_T$ itself is a tree, namely the $1$-subdivision of $T$ with one extra pendant per original vertex. Note that our construction is similar to the one described in~\cite[Theorem 11]{KOW25} where we attach one pendent instead of $\ell$ per original vertex.

\paragraph*{Soundness and Completeness}
 
We prove two complementary bounds connecting the bandwidth of $T$ and the smallest matching-planarity parameter of $G_T$. Combined with the constant-factor inapproximability of bandwidth, they yield Theorem \ref{thm:inapprox}. Following similar arguments as in~\cite[Theorem 11]{KOW25} we prove the following lemmas.

\begin{lemma}\label{lem:bw-to-kmatching}
    If $\operatorname{bw}(T) \le b$, then $G_T$ admits a $(3b-3)$-matching-planar~drawing.
\end{lemma}

\begin{proof}
Suppose that $T$ has bandwidth at most $b$ for some integer $b\geq 1$, and let $\sigma=(v_1,\dots,v_n)$ be a corresponding linear order. We claim that $G_T$ admits a $2$-layer $(3b-3)$-matching-planar drawing. Define $\prec_X$ on $X_T$ to agree with $\sigma$, that is, $v_1\prec_X v_2\prec_X\dots\prec_X v_n$. For $\prec_Y$, place the pendants in the same order as their parents, $y_{v_1}\prec_Y y_{v_2}\prec_Y\dots\prec_Y y_{v_n}$. For each tree edge $uv\in E(T)$ with $u=v_i$, $v=v_j$ and $i<j$, place $w_{uv}$ in $\prec_Y$ between $y_{v_i}$ and $y_{v_j}$ so that approximately half of the intermediate pendants $y_{v_{i+1}},\dots,y_{v_{j-1}}$ lie to its left and the other half to its right.
 
We check that every edge of $G_T$ has its set of crossings covered by a matching of size at most $3b-3$, considering each edge type:
\begin{itemize}
    \item Pendant edge $v_i y_{v_i}$: Because the pendants are ordered along $\prec_Y$ in the same way as their parents along $\prec_X$, no pendant edge crosses another pendant edge. The only edges that can cross $v_i y_{v_i}$ are subdivision edges $v_p w_{e'}$ for some tree edge $e'=uv\in E(T)$. Such a subdivision edges crosses $v_i y_{v_i}$ exactly when $\sigma(u)<i<\sigma(v)$ or $\sigma(v)<i<\sigma(u)$. By the bandwidth bound, both endpoints of any crossing tree edge $uv$ lie within $\sigma$-distance $b-1$ of $v_i$. The crossing tree edges therefore form a forest on the at most $2b-2$ vertices $v_t$ with $0<|\sigma(v_t)-i|\leq b-1$, hence number at most $2b-3$. Each contributes at most one subdivision edges (the two subdivision edges share $w_{e'}$, so they cannot both have their $X$-endpoint on the appropriate side of $v_i$). The contributing subdivision edges are pairwise non-adjacent, since distinct crossing tree edges give distinct $X$- and $Y$-endpoints. So the matching number of the crossings is at most $2b-2\leq 3b-3$.

    \item Subdivision edge $v_iw_{uv}$ where $uv\in E(T)$ with $u=v_i$, $v=v_j$, and $i<j$: The crossing edges fall into two classes. First, the pendant edges $v_t y_{v_t}$ with $i<t<j$ whose pendant $y_{v_t}$ lies on the opposite side of $w_{uv}$ in $\prec_Y$. By the bandwidth bound there are at most $j-i-1\leq b-1$ such intermediate vertices; even if all of them contribute, they form a matching of size at most $b-1$. Second, the subdivision edges $v_p w_{e'}$ of other tree edges $e'\neq uv$ that cross $v_i$ in $\sigma$. By the same forest argument as above, there are at most $2b-2$ such tree edges, each contributing at most one crossing subdivision edge, all pairwise non-adjacent. Hence the matching number of the crossings is at most $(b-1)+(2b-2)=3b-3$.
\end{itemize}
Therefore the drawing is $(3b-3)$-matching-planar.
\end{proof}
\begin{lemma}\label{lem:kmatching-to-bw}
    If $G_T$ is $k$-matching-planar, then $\operatorname{bw}(T)\le 2k + 1$.
\end{lemma}
\begin{proof}
Suppose that $G_T$ admits a $2$-layer $k$-matching-planar drawing for some integer $k\geq 0$, with linear orders $\prec_X^\ast$ on $X_T$ and $\prec_Y^\ast$ on $Y_T$. Let $\sigma^\ast$ be the order on $V(T)=X_T$ induced by $\prec_X^\ast$. We claim that the bandwidth of $T$ with respect to $\sigma^\ast$ is at most $2k+1$, which then implies $\operatorname{bw}(T)\leq 2k+1$.
 
Suppose for the sake of contradiction that some tree edge $uv\in E(T)$ has stretch $d\geq 2k+2$ in $\sigma^\ast$. Without loss of generality, $u\prec_X^\ast v$, and there are exactly $d-1\geq 2k+1$ vertices of $V(T)$ strictly between $u$ and $v$ in $\prec_X^\ast$; denote them by $v_{t_1},\dots,v_{t_{d-1}}$ in their order along $\prec_X^\ast$.
 
Consider the two halves $uw_{uv}$ and $vw_{uv}$ of the subdivided edge. For each intermediate vertex $v_{t_s}$, the pendant edge $v_{t_s}y_{v_{t_s}}$ has its $X$-endpoint strictly between $u$ and $v$ in $\prec_X^\ast$, while its $Y$-endpoint $y_{v_{t_s}}$ lies somewhere relative to $w_{uv}$. The edge $v_{t_s} y_{v_{t_s}}$ crosses exactly one of $uw_{uv}$ and $vw_{uv}$.
By the pigeonhole principle, at least $\lceil (d-1)/2\rceil\geq k+1$ of the pendant edges $v_{t_s}y_{v_{t_s}}$ cross the same subdivision edge, say $uw_{uv}$ without loss of generality. These pendant edges are pairwise non-adjacent, so they form a matching of size at least $k+1$ in the set of edges crossing $uw_{uv}$. This contradicts the $k$-matching-planarity of the drawing. Therefore every tree edge has stretch at most $2k+1$ in $\sigma^\ast$, hence $\operatorname{bw}(T)\leq 2k+1$.
\end{proof}

To complete the proof of the theorem, suppose for the sake of contradiction that there is a polynomial-time $c$-approximation algorithm $\mathcal{A}$ for the smallest such $k$, for some constant~$c\geq 1$. Given a caterpillar $T$, construct $G_T$ in polynomial time. Let $k$ be the smallest integer such that $G_T$ admits a $k$-matching-planar drawing. We run $\mathcal{A}$ on $G_T$ to obtain $\hat{k}$ with $k\leq \hat{k}\leq c k$, and output $\hat{b}=2\hat{k}+1$. By \Cref{lem:kmatching-to-bw}, $\operatorname{bw}(T)\leq 2k+1\leq 2\hat{k}+1=\hat{b}$. By \Cref{lem:bw-to-kmatching},
\[
\hat{b}=2\hat{k}+1\leq 2c k+1\leq 2c(3\operatorname{bw}(T)-3)+1\leq 6c\cdot\operatorname{bw}(T).
\]
Hence $\operatorname{bw}(T)\leq\hat{b}\leq 6c\cdot\operatorname{bw}(T)$, so $\hat{b}$ is a $6c$-approximation of $\operatorname{bw}(T)$ computable in polynomial time. As $6c$ is a constant, this contradicts the constant-factor inapproximability of bandwidth on caterpillars~\cite{DFU11}.
\end{proof}

\section{Future Work}

Several natural questions remain open. 
\begin{question}
    Does \textsc{TwoSided-$k$-MatchingPlanarity} admit an $\operatorname{XP}$-algorithm parameterized by~$k$? Is it $\operatorname{XNLP}$-hard, as is the case for $2$-layer $k$-planar graphs~\cite{KOW25}?
\end{question}

\begin{question}
    What is the exact maximum pathwidth of $2$-layer $k$-matching-planar graphs, closing the current gap between the lower bound $3\lfloor k/2\rfloor + 1$ and the upper bound $2k+1$?
\end{question}

\emph{$k$-cover-planar graphs} have been introduced by Hendrey et al.~\cite{HKW25}, where the vertex cover number (rather than the matching number) of the edges crossing any given edge is at most~$k$. By K\H{o}nig's theorem, these classes coincide in the $2$-layer setting. The general edge density bound of $3e(k+1)n$ for $k$-cover-planar graphs~\cite[Lemma~5.1]{HKW25} is therefore improved by our pathwidth bound to $(2k+1)n$, while $K_{k+1,n-k-1}$ is $2$-layer $k$-matching-planar and has $\approx(k+1)n$ edges.

\begin{question}
What is the tight edge density of $2$-layer $k$-matching-planar graphs?
\end{question}

A graph is \emph{circular $k$-matching-planar} if it admits a straight-line drawing with vertices on a circle in which the matching number of the edges crossing any given edge is at most~$k$. Any $2$-layer $k$-matching-planar drawing can be interpreted as a circular drawing by treating the two layers as two consecutive arcs of a circle. Note that pathwidth is unbounded for this class, as outerplanar graphs are circular $0$-matching-planar and there are $n$-vertex outerplanar graphs of pathwidth $\Omega(\log n)$ (e.g. complete binary trees). Hendrey et al.~\cite{HKW25} proved that graphs with such a drawing have treewidth~$O(k^3\log^2 k)$.

\begin{question}
What is the optimal upper bound on the treewidth of circular $k$-matching-planar graphs?
\end{question}

\begin{question}
What is the complexity of recognizing circular $k$-matching-planar graphs?
\end{question}

Finally, a key technical ingredient for developing an $\operatorname{XP}$-algorithm for recognizing circular $k$-planar graphs and for obtaining a tight upper bound on their treewidth is the following \emph{triangulation lemma}~\cite[Lemma~6]{FGKO024}: given a circular $k$-planar drawing of a graph $G$, there exists a triangulation $T$ of the outer cycle of $G$ such that every edge of $T$ is \emph{pierced} by at most $k$ edges of~$G$. This lemma underlies both the treewidth bound proved in~\cite{FGKO024} and the $\operatorname{XP}$-algorithm for recognizing circular $k$-planar graphs~\cite{KOW25}. We therefore ask the following question.

\begin{question}
    Is it true that, for every $k \ge 1$, every circular $k$-matching-planar drawing admits a triangulation of the outer cycle such that no $k+1$ pairwise independent edges of the drawing pierce any edge of the triangulation?
\end{question}

\bibliography{literature}

\end{document}